\documentclass{article}

\usepackage[utf8]{inputenc}
\usepackage{amsthm,mathtools}
\usepackage[margin=1in]{geometry}
\usepackage{amsfonts}
\usepackage{color,soul}
\usepackage{amssymb} 
\usepackage{footmisc}
\usepackage[T1]{fontenc}
\usepackage{lmodern}
\usepackage{hyperref}
\usepackage{algpseudocode}
\usepackage{algorithm}
\usepackage{float}
\newtheorem{theorem}{Theorem}[section]
\newtheorem{lemma}[theorem]{Lemma}
\newtheorem{proposition}[theorem]{Proposition}
\newtheorem{fact}[theorem]{Fact}
\newtheorem{corollary}[theorem]{Corollary}

\theoremstyle{definition}
\newtheorem{definition}[theorem]{Definition}

\newcommand{\ds}{\text{\textcircled{s}}}
\newcommand{\R}{\mathbb{R}}
\newcommand{\Q}{\mathbb{Q}}

\newcommand{\class}[1]{\mathbf{#1}}

\newcommand{\BPP}{\class{BPP}}

\newcommand{\RP}{\class{RP}}

\newcommand{\RL}{\class{RL}}
\newcommand{\SL}{\class{SL}}
\renewcommand{\L}{\class{L}}
\newcommand{\tL}{\tilde{\class{L}}}

\newcommand{\BPL}{\class{BPL}}

\newcommand{\SPACE}{\class{SPACE}}

\renewcommand{\P}{\class{P}}
\newcommand{\eqdef}{\mathbin{\stackrel{\rm def}{=}}}

\newcommand{\tO}{\tilde{O}}
\newcommand{\poly}{\mathrm{poly}}
\newcommand{\polylog}{\mathrm{polylog}}

\def\textprob#1{\textmd{\textsc{#1}}}
\newcommand{\USTConn}{\textprob{Undirected S-T Connectivity}}

\newcommand{\ULapSys}{\textprob{Undirected Laplacian Systems}}

\title{Derandomization Beyond Connectivity:\\Undirected Laplacian Systems in Nearly Logarithmic Space}
\date{\today}
\author{Jack Murtagh\thanks{Supported by NSF grant CCF-1420938} \\School of Engineering \& Applied Sciences\\
Harvard University\\ 
Cambridge, MA USA\\
\texttt{jmurtagh@g.harvard.edu}\\
\url{http://scholar.harvard.edu/jmurtagh} \and Omer Reingold \\ Computer Science Department\\ Stanford University\\ Stanford, CA USA \\ \texttt{reingold@stanford.edu}  \and Aaron Sidford \\ Management Science \& Engineering \\ Stanford University \\ Stanford, CA USA \\ \texttt{sidford@stanford.edu} \\ \url{http://www.aaronsidford.com} \and 
Salil Vadhan\thanks{Supported by NSF grant CCF-1420938 and a Simons Investigator Award.}\\School of Engineering \& Applied Sciences\\
Harvard University\\ 
Cambridge, MA USA\\
\texttt{salil@seas.harvard.edu}\\
\url{http://seas.harvard.edu/~salil}}

\begin{document}

\begin{titlepage}
\maketitle

\begin{abstract}
We give a deterministic $\tO(\log n)$-space algorithm for approximately solving linear systems given by Laplacians of undirected graphs, and consequently also approximating hitting times, commute times, and escape probabilities for undirected graphs.  Previously, such systems were known to be solvable by randomized algorithms using $O(\log n)$ space (Doron, Le Gall, and Ta-Shma, 2017) and hence by deterministic algorithms using  $O(\log^{3/2} n)$ space (Saks and Zhou, FOCS 1995 and JCSS 1999).  

Our algorithm combines ideas from time-efficient Laplacian solvers (Spielman and Teng, STOC `04; Peng and Spielman, STOC `14) with ideas used to show that \USTConn \ is in deterministic logspace (Reingold, STOC `05 and JACM `08; Rozenman and Vadhan, RANDOM `05).
\end{abstract}

\vfill
\end{titlepage}

\begin{section}{Introduction}

\subsection{The RL vs. L Problem}

One of the central problems in computational complexity is to understand the power of randomness for efficient computation.  There is strong evidence that randomness does not provide a substantial savings for standard algorithmic problems, where the input is static and the algorithm has time to access it in entirety.   Indeed, under widely believed complexity assumptions (e.g. that SAT has circuit complexity $2^{\Omega(n)}$), it is known that every randomized algorithm can be made deterministic with only a polynomial slowdown ($\BPP=\P$, $\RP=\P$) and a constant-factor increase in space ($\BPL=\L$, $\RL=\L$)~\cite{ImpagliazzoWi97,KlivansMe02}.  

A major challenge is to prove such derandomization results unconditionally, without relying on unproven complexity assumptions.  In the time-bounded case, it is known that proving that $\RP=\P$ requires progress on circuit lower bounds~\cite{ImpagliazzoKaWi02,KabanetsIm04}.
But for the space-bounded case, there are no such barriers, and we can hope for an
unconditional proof that $\RL=\L$.   Indeed, Nisan~\cite{Nisan92} gave an unconditional construction of a pseudorandom generator for space-bounded computation with seed length $O(\log^2 n)$, which was used
by Saks and Zhou~\cite{SaksZh99} to prove that 
$\RL\subseteq \SPACE(\log^{3/2} n) \eqdef \L^{3/2}$.  Unfortunately,
despite much effort over the past two decades, this remains the best known upper bound on the deterministic space complexity of $\RL$.

For many years, the most prominent example of a problem in $\RL$ not known to be in $\L$ was \USTConn, which can be solved in randomized logarithmic space by performing a polynomial-length random walk from the start vertex $s$ and accepting if the destination vertex $t$ is ever visited~\cite{AleliunasKaLiLoRa79}.   In 2005, Reingold~\cite{Reingold08} gave a deterministic logspace algorithm for this problem.  Since \USTConn\ is complete for $\SL$ (symmetric logspace)~\cite{LewisPa80}, this also implies deterministic logpsace algorithms for several other natural problems, such as \textprob{Bipartiteness}~\cite{JonesLiLa76,AlvarezGr00}. 

Reingold's algorithm provided hope for renewed progress on the general $\RL$ vs. $\L$ problem.  Indeed, it was shown in \cite{ReingoldTrVa06} that solving S-T connectivity on {\em directed} graphs where the random walk is promised to have polynomial mixing time (and where $t$ has noticeable stationary probability) is complete for $\RL$ (generalized to promise problems).   While \cite{ReingoldTrVa06} was also able to generalize Reingold's methods to Eulerian directed graphs, the efforts to handle all of $\RL$ with this approach stalled.  Thus, researchers turned back to constructing pseudorandom generators for more restricted models of space-bounded computation, such as various types of constant-width branching programs~\cite{BravermanRaRaYe10,BrodyVe10, KouckyNiPu11,SimaZa11,GopalanMeReTrVa12,ImpagliazzoMeZu12,Steinke12,ReingoldStVa13,SteinkeVaWa14}.

\subsection{Our Work}
In this paper, we restart the effort to obtain deterministic logspace algorithms for increasingly rich graph-theoretic problems beyond those known to be in $\SL$.  In particular, we provide a nearly logarithmic space algorithm for approximately solving \ULapSys, which also implies nearly logarithmic-space algorithms for approximating hitting times, commute times, and escape probabilities for random walks on undirected graphs.  

Our algorithms are obtained by combining the techniques from recent
nearly linear-time Laplacian solvers~\cite{ST04,PS13} with methods related to Reingold's algorithm~\cite{RozenmanVa05}.  The body of work of time-efficient Laplacian solvers has recently been extended to Eulerian directed graphs~\cite{CKPPRSV16}, which in turn was used to obtain algorithms to approximate stationary probabilities for arbitrary directed graphs with polynomial mixing time through recent reductions \cite{CKPPSV16}.

This raises the tantalizing possibility of extending our nearly logarithmic-space Laplacian solvers in a similar way to prove that $\RL\subseteq \tL \eqdef \SPACE(\tO(\log n))$, since approximating stationary probabilities on digraphs with polynomial mixing time suffices to solve all of $\RL$~\cite{ReingoldTrVa06,ChungReVa11}, and Reingold's algorithm has been extended to Eulerian digraphs~\cite{ReingoldTrVa06,RozenmanVa05}.

\begin{subsection}{Laplacian system solvers} 
Given an undirected multigraph $G$ with adjacency matrix $A$ and diagonal degree matrix $D$, the {\em Laplacian} of $G$ is the matrix $L=D-A$. Solving systems of equations $Lx=b$ in the Laplacians of graphs (and in general symmetric diagonally dominant systems) arises in many applications throughout computer science. In a seminal paper, Spielman and Teng gave a nearly linear time algorithm for solving such linear systems \cite{ST04}, which was improved in a series of follow up works \cite{KMP10, KMP11, KOSZ13, PS13, LS13} including recent extensions to directed graphs \cite{CKPPSV16, CKPPRSV16}. These methods have sparked a large body of work using Laplacian solvers as an algorithmic primitive for problems such as max flow \cite{CKMST11, KLOS14}, randomly sampling spanning trees \cite{KM09}, sparsest cut \cite{S09}, graph sparsification \cite{SS11}, as well as problems in computer vision \cite{KMT09}.

Recent works have started to study the space complexity of solving Laplacian systems. Ta-Shma~\cite{T13} gave a {\em quantum} logspace algorithm for approximately solving general linear systems that are suitably well-conditioned.  Doron, Le Gall, and Ta-Shma~\cite{DLT17} showed that there is a {\em randomized} logspace algorithm for approximately solving Laplacian systems on digraphs with polynomial mixing time. 
\end{subsection}

\subsection{Main Result}

We give a nearly logarithmic-space algorithm for approximately solving undirected Laplacian systems:

\begin{theorem} \label{thm:main}
There is a deterministic algorithm that, given an undirected multigraph $G=(V,E)$ specified as a list of edges, a vector $b\in \Q^{|V|}$ in the image of $G$'s Laplacian $L$, and an approximation parameter $\epsilon>0$, finds a vector $x\in \Q^{|V|}$, such that $$\|x-x^*\|\leq \epsilon\cdot \|x^*\|$$ for some $x^*$ such that $Lx^*=b$, in space
$$O\left(\log n\cdot\log\log (n/\epsilon)\right),$$
where $n$ is the bitlength of the input $(G,b,\epsilon)$. 
\end{theorem}
In particular, for $\epsilon=1/\poly(n)$, the algorithm uses space
$\tO(\log n)$.  Note that since the algorithm applies to multigraphs (represented as edge lists), we can handle polynomially large integer edge weights. Known reductions~\cite{CKPPSV16} imply the same space bounds for estimating hitting times, commute times, and escape probabilities for random walks on undirected graphs.  (See Section~\ref{sec:applications}.)

\subsection{Techniques}

The starting point for our algorithm is the undirected Laplacian solver of Peng and Spielman~\cite{PS13}, which is a randomized algorithm that uses polylogarithmic parallel time  and a nearly linear amount of work.   (Interestingly,
concurrently and independently of \cite{Reingold08}, Trifonov~\cite{trifonov05} gave a $\tO(\log n)$-space algorithm for \USTConn\ by importing techniques from parallel algorithms, like we do.)  It implicitly\footnote{Nearly linear-time algorithms for Laplacian system solvers do not have enough time to write done a full approximate pseudoinverse, which may be dense; instead they compute the result of applying the approximate matrix to a given vector.} 
computes an approximate {\em pseudoinverse} $L^{+}$ of a graph Laplacian (formally defined in Section~\ref{sec:pseudoinverse}), which is equivalent to approximately solving Laplacian systems. Here we will sketch their algorithm and how we obtain a space-efficient analogue of it.

By using the deterministic logspace algorithm for \USTConn~\cite{Reingold08}, we may assume without loss of generality that our input graph $G$ is connected (else we find the connected components and work on each separately).  By adding self-loops (which does not change the Laplacian), we may also assume that $G$ is regular and nonbipartite.  For notational convenience, here and through most of the paper we will work with the {\em normalized Laplacian}, which in the case that $G$ is regular,\footnote{In an irregular graph, the normalized Laplacian is defined to be $I-D^{-1/2}AD^{-1/2}$, where $D$ is the diagonal matrix of vertex degrees.  When $G$ is $d$-regular, we have $D=dI$, so $D^{-1/2}AD^{-1/2} = A/d=M$.} equals $L=L(G)=I-M$, where  $M=M(G)$ is the $n\times n$ transition matrix for the random walk on $G$ and now we are using $n$ for the number of vertices in $G$. 
Since the uniform distribution is stationary for the random walk on a regular graph, the all-ones vector $\vec{1}$ is in the kernel of $L$.  Because $G$ is connected, there is no other stationary distribution and $L$ is of rank $n-1$. Thus computing $L^+$ amounts to inverting $L$ on the space orthogonal to $\vec{1}$. 

The Peng--Spielman algorithm is based on the following identity: 
$$(I-M)^+ = \frac{1}{2}\cdot\left(I-J + (I+M)(I-M^2)^+(I+M)\right),$$
where $J$ is the matrix with all entries $1/n$.   That is,
\begin{equation}
\label{eqn:pengspielman}
L(G)^+ = \frac{1}{2}\cdot\left(I-J+(I+M(G))L(G^2)^+(I+M(G))\right),
\end{equation}
where $G^2$ is the multigraph whose edges correspond to walks of length 2 in $G$.

This gives rise to a natural algorithm for computing $L(G)^+$, by recursively computing $L(G^2)^+$ and applying Equation (\ref{eqn:pengspielman}).  
After squaring the graph $k=O(\log n)$ times, we are considering all walks of length $2^k$, which is beyond the $O(n^2)$ mixing time of regular graphs, and hence $M(G^{2^k})$ is approximately $J$, and we can easily approximate the pseudoinverse as $L(G^{2^k})^+ 
\approx (I-J)^+ = I-J$, the Laplacian of the complete graph with self loops on every vertex.

However, each time we square the graph, the degree also squares, which is too costly in terms of computational complexity.  Thus, to obtain nearly linear time, Peng and Spielman~\cite{PS13} sparsify $G^2$ at each step of recursion.   Specifically, they carefully combine elements of randomized sparsification procedures from \cite{ST11,OV11,ST13} that retains $\tO(n)/\epsilon^2$ edges from $G^2$ and provides a spectral $\epsilon$-approximation $G'$ to $G^2$ in the sense that for every vector $v$, we have $v^T L(G') v = (1\pm \epsilon) v^T L(G^2) v$. The recursion of Equation~(\ref{eqn:pengspielman}) behaves nicely with respect to spectral approximation, so that if we replace $G^2$ with such an $\epsilon$-approximation $G'$ at each step, we obtain a $O(k\epsilon)$-approximation over the $k=O(\log n)$ levels of recursion provided $\epsilon\leq 1/O(k)$. Thus, we can
take $\epsilon=1/O(\log n)$ and obtain a good spectral approximation overall, using graphs with $\tO(n)$ edges throughout and thereby maintaining a nearly linear amount of work.

There are two issues with trying to obtain a deterministic $\tO(\log n)$-space algorithm from this approach.  First, we cannot perform random sampling to sparsify.  Second, even if we derandomize the sampling step in deterministic logarithmic space, a direct recursive implementation would cost $O(\log n)$ space for each of the $k=O(\log n)$ levels of recursion, for a space bound of $O(\log^2 n)$.

For the first issue, we show that the sparsification can be done deterministically using the {\em derandomized square} of Rozenman and Vadhan~\cite{RozenmanVa05}, which was developed in order to give a simpler proof that \USTConn\ is in $\L$.  If $G$ is a $d$-regular graph, a derandomized square of $G$ is obtained by using an expander $H$ on $d$ vertices to select a subset of the walks of length 2 in $G$.  If $H$ has degree $c\ll d$, then the corresponding derandomized square of $G$ will have degree $d\cdot c \ll d^2$.   In \cite{RozenmanVa05}, it was shown that the derandomized square improves expansion (as measured by spectral gap) nearly as well actual squaring, to within a factor of $1-\epsilon$, where $\epsilon$ is the second largest eigenvalue of $H$ in absolute value.   Here we show that derandomized squaring actually provides a spectral $O(\epsilon)$-approximation to $G^2$. Consequently, we can use derandomized squaring with expanders of second eigenvalue $1/O(\log n)$ and hence degree $c=\polylog(n)$ in the Peng--Spielman algorithm. 

Next, to obtain a space-efficient implementation, we consider what happens when we repeatedly apply Equation~(\ref{eqn:pengspielman}), but with true squaring replaced by derandomized squaring.  We obtain a sequence of graphs $G_0, G_1,\ldots,G_k$, where $G_0=G$ and $G_i$ is the derandomized square of $G_{i-1}$, a graph of degree $d\cdot c^i$.  We recursively compute spectral approximations  of the pseudoinverses $L(G_i)^+$ as follows:
\begin{equation}
\label{eqn:ourrecursion}
\widetilde{L(G_{i-1})^+} = \frac{1}{2}\cdot\left(I-J+(I+M(G_{i-1}))\widetilde{L(G_i)^+}(I+M(G_{i-1}))\right).
\end{equation}
Opening up all $k=O(\log n)$ levels of this recursion, there is
an explicit {\em quadratic} matrix polynomial $p$ where each of the (noncommuting) variables appears at most twice in each term, such that
$$\widetilde{L(G_0)} = p(I+M(G_0),I+M(G_1),\ldots,I+M(G_k)).$$ 
Our task is to evaluate this polynomial in space $\tO(\log n)$.  
First, we use the fact, following \cite{Reingold08,RozenmanVa05}, that a careful implementation of $k$-fold derandomized squaring using explicit expanders of degree $c$ can be evaluated in space $O(\log n+k\cdot\log c)=\tO(\log n)$.  (A naive evaluation would take space $O(k\cdot \log n)=O(\log^2 n)$.)  This means that we can construct each of the matrices $I+M(G_{i})$ in space $\tO(\log n)$.\footnote{We follow the usual model of space-bounded computation, where algorithms can have output (to write-only memory) that is larger than their space bound.  We review the standard composition lemma for such algorithms in Section~\ref{sect:spacemodel}.}  Next, we observe that each term in the multilinear polynomial multiplies at most $2k+1$ of these matrices together, and hence can be computed recursively in space $O(\log k)\cdot O(\log n)=\tO(\log n)$.   Since iterated addition is also in logspace, we can then sum to evaluate the entire matrix polynomial in space $\tO(\log n)$.

To obtain an arbitrarily good final approximation error $\epsilon>0$, we could use expanders of degree $c=\poly((\log n)/\epsilon)$ for our derandomized squaring, but this would yield a space complexity of at least $k\cdot \log n \geq \log(1/\epsilon)\cdot \log n$.  To obtain the doubly-logarithmic dependence on $1/\epsilon$ claimed in Theorem~\ref{thm:main}, we follow the same approach as \cite{PS13}, computing a constant factor spectral approximation as above, and then using ``Richardson iterations'' at the end to improve the approximation factor.  Interestingly, this doubly-logarithmic dependence on $\epsilon$ is even better than what is achieved by the randomized algorithm of \cite{DLT17}, which uses space $O(\log(n/\epsilon))$.  
\end{section}

\begin{section}{Preliminaries}
\begin{subsection}{Graph Laplacians}
Let $G=(V,E)$ be an undirected multigraph on $n$ vertices. The {\em adjacency matrix} $A$ is the $n\times n$ matrix such that entry $A_{ij}$ contains the number of edges from vertex $j$ to vertex $i$ in $G$. The {\em degree matrix} $D$ is the diagonal $n\times n$ matrix such that $D_{ii}$ is the degree of vertex $i$ in $G$. The {\em Laplacian} of $G$ is defined to be
\[
L=D-A.
\]
$L$ is a symmetric, positive semidefinite matrix with non-positive off-diagonal entries and diagonal entries $L_{ii}=\sum_{j\not=i}|L_{ij}|$ equaling the sum of the absolute values of the off-diagonal entries in that row. Each row of a Laplacian sums to 0 so the all-ones vector is in the kernel of every Laplacian.  

We will also work with the normalized Laplacian of regular graphs $LD^{-1}=I-M$ where $M=AD^{-1}$ is the transition matrix of the random walk on the graph $G$.  

\begin{definition}
Let $G$ be a regular undirected graph with transition matrix $M$. Then we define 
\[
\lambda(G)=\max_{\substack{v\perp\vec{1}\\ v\neq 0}}\frac{\|Mv\|}{\|v\|}=\mathrm{2nd~largest ~absolute~value ~of ~the~ eigenvalues~ of ~}M.
\]
The {\em spectral gap} of $G$ is $\gamma(G)=1-\lambda(G)$
\end{definition}
The spectral gap of a multigraph is a good measure of its expansion properties. It is shown in \cite{RozenmanVa05} that if $G$ is a $d$-regular multigraph on $n$ vertices with a self loop on every vertex then $\lambda(G)\leq 1-1/2d^{2}n^{2}$.

Throughout the paper we let $J$ denote the $n\times n$ matrix with $1/n$ in every entry. So $I-J$ is the normalized Laplacian of the complete graph with a self loop on every vertex. We use $\vec{1}$ to denote the all ones vector.
\end{subsection}

\begin{subsection}{Moore-Penrose pseudoinverse} \label{sec:pseudoinverse}
Since a normalized Laplacian $L=I-M$ is not invertible, we will consider the Moore-Penrose pseudoinverse of $L$. 
\begin{definition}
Let $L$ be a real-valued matrix. The {\em Moore-Penrose pseudoinverse} $L^{+}$ of $L$ is the unique matrix satisfying the following:
\begin{enumerate}
\item $LL^{+}L=L$
\item $L^{+}LL^{+}=L^{+}$
\item $L^{+}L = (L^{+}L)^{\intercal}$
\item $LL^{+} = (LL^{+})^{\intercal}$.
\end{enumerate}
\end{definition}

\begin{fact}
For a real-valued matrix $L$, $L^+$ has the following properties:
\begin{enumerate}
\item $\mathrm{ker}(L^{+})=\mathrm{ker}(L^{\intercal})$
\item $\mathrm{im}(L^{+})=\mathrm{im}(L^{\intercal})$
\item If $c$ is a nonzero scalar then $(cL)^{+}=c^{-1}L^{+}$
\item $L^+$ is real-valued
\item If $L$ is symmetric with eigenvalues $\lambda_1,\ldots,\lambda_n$, then $L^{+}$ has the same eigenvectors as $L$ and eigenvalues $\lambda_1^{+},\ldots,\lambda_n^{+}$ where
\[
\lambda_i^{+}=\begin{dcases}
\frac{1}{\lambda_i} ~&\mathrm{if~} \lambda_i\not=0\\ 
0 ~&\mathrm{if~} \lambda_i=0.
\end{dcases}
\]
\end{enumerate}
\end{fact}
Next we show that solving a linear system in $L$ can be reduced to computing the pseudoinverse of $L$ and applying it to a vector.
\begin{proposition}
\label{prop:pseudoinversesolution}
If $Lx=b$ has a solution, then the vector $L^{+}b$ is a solution. 
\end{proposition}
\begin{proof}
Let $x^{*}$ be a solution. Multiplying both sides of the equation  $Lx^{*}=b$ by $LL^{+}$ gives 
\begin{align*}
    LL^{+}Lx^{*}&=LL^{+}b\\
    \implies Lx^{*}&=LL^{+}b\\
    \implies b&=LL^{+}b\\
\end{align*}
The final equality shows that $L^{+}b$ is a solution to the system. 
\end{proof}
\end{subsection}

\begin{subsection}{Spectral approximation}
We want to approximate a solution to $Lx=b$.  Our algorithm works by computing an approximation $\tilde{L}^{+}$ to the matrix $L^{+}$ and then outputting $\tilde{L}^{+}b$ as an approximate solution to $Lx=b$. We use the notion of spectral approximation of matrices first introduced in \cite{ST11}.

\begin{definition}
\label{def:spectralapprox}
Let $X, Y$ be symmetric $n\times n$ real matrices. We say that $X\approx_{\epsilon}Y$ if
\[
e^{-\epsilon}\cdot X \preceq Y \preceq e^{\epsilon}\cdot X
\]
where for any two matrices $A,B$ we write $A\preceq B$ if $B-A$ is positive semidefinite.
\end{definition}

Below are some additional useful facts about spectral approximation that we use throughout the paper.

\begin{proposition}[\cite{PS13}]
\label{prop:psdfacts}
If $X, Y,W,Z$ are positive semidefinite $n\times n$ matrices then 
\begin{enumerate}
\item If $X\approx_{\epsilon}Y$ then $Y\approx_{\epsilon}X$
\item If $X\approx_{\epsilon}Y$ and $c\geq 0$ then $cX\approx_{\epsilon}cY$ 
\item If $X\approx_{\epsilon}Y$ then $X+W\approx_{\epsilon}Y+W$
\item If $X\approx_{\epsilon}Y$ and $W\approx_{\epsilon}Z$ then $X+W\approx_{\epsilon}Y+Z$
\item If $X\approx_{\epsilon_1}Y$ and $Y\approx_{\epsilon_2}Z$ then $X\approx_{\epsilon_1 + \epsilon_2}Z$
\item If $X\approx_{\epsilon}Y$ and $M$ is any matrix then $M^{\intercal}XM\approx_{\epsilon}M^{\intercal}YM$
\item If $X\approx_{\epsilon}Y$ and $X$ and $Y$ have the same kernel, then $X^{+}\approx_{\epsilon}Y^{+}$
\item If $X\approx_{\epsilon}Y$ then $I\otimes X\approx_{\epsilon} I\otimes Y$ where $I$ is the identity matrix (of any dimension) and $\otimes$ denotes the tensor product.
\end{enumerate}
\end{proposition}
For completeness, we include the proof of Proposition \ref{prop:psdfacts} in Appendix \ref{app:psdfacts}.
\end{subsection}
\end{section}

\begin{section}{Main Theorem and Approach}
Our Laplacian solver works by approximating $L^{+}$. The main technical result is: 
\begin{theorem}
\label{thm:arbitraryapproxmain}
Given an undirected, connected multigraph $G$ with Laplacian $L=D-A$ and $\epsilon>0$, there is a deterministic algorithm that computes $\tilde{L}^{+}$ such that $\tilde{L}^{+}\approx_{\epsilon}L^{+}$ that uses space $O(\log n\cdot\log\log(n/\epsilon))$ where $n$ is the bitlength of the input.
\end{theorem}
Approximating solutions to Laplacian systems follows as a corollary and is discussed in Section \ref{sec:applications}. To prove Theorem \ref{thm:arbitraryapproxmain} we first reduce the problem to the task of approximating the pseudoinverse of the normalized Laplacian of a regular, aperiodic, multigraph with degree a power of 2. Let $L=D-A$ be the Laplacian of an undirected multigraph $G$. We can make $G$ $f$-regular with $f$ a power of 2 and aperiodic by adding an appropriate number of self loops to every vertex. Let $E$ be the diagonal matrix of self loops added to each vertex in $G$. Then $L=D+E-A-E$ is the Laplacian of the regular, aperiodic multigraph (self loops do not change the unnormalized Laplacian) and $L(D+E)^{-1}=L/f$ is the normalized Laplacian. Recalling that $(L/f)^+=f\cdot L^{+}$ completes the reduction. 

Our algorithm for computing the pseudoinverse of the normalized Laplacian of a regular, aperiodic multigraph is based on the Laplacian solver of Peng and Spielman \cite{PS13}. It works by using the following identity:
\begin{proposition}
\label{prop:identity}
If $L=I-M$ is the normalized Laplacian of an undirected, connected, regular, aperiodic, multigraph $G$ on $n$ vertices then
\[
L^{+}=\frac{1}{2}\left(I-J+(I+M)(I-M^2)^{+}(I+M)\right)
\]
\end{proposition}
This identity is adapted from the one used in \cite{PS13}. A proof of its correctness can be found in Appendix \ref{app:identity}.

Recall that squaring the transition matrix $M$ of a regular multigraph $G$ yields the transition matrix of $G^{2}$, which is defined to be the graph on the same vertex set as $G$ whose edges correspond to walks of length 2 in $G$. So the identity from Proposition \ref{prop:identity} reduces the problem of computing the pseudoinverse of the normalized Laplacian of $G$ to computing the pseudoinverse of the normalized Laplacian of $G^{2}$ (plus some additional matrix products). Repeatedly applying the identity from Proposition \ref{prop:identity} and expanding the resulting expression we see that for all integers $k\geq 0$
\begin{align}
\label{idexpansion}
L^{+}&= \frac{1}{2}(I-J)+\frac{1}{2}(I+M)(I-M^2)^{+}(I+M)\\
&=\frac{1}{2}(I-J)+\frac{1}{4}(I+M)(I-J)(I+M)+\frac{1}{4}(I+M)(I+M^2)(I-M^4)^{+}(I+M^2)(I+M) \\
&=\frac{1}{2}(I-J)+\left(\sum_{i=0}^{k-1}\frac{1}{2^{i+2}}W_{i}\right)+\frac{1}{2^{k+1}}(I+M)\cdot\ldots\cdot(I+M^{2^{k}})(I-M^{2^{k+1}})^{+}(I+M^{2^{k}})\cdot\ldots\cdot(I+M)
\end{align}
where for all $i\in \{0,\ldots,k-1\}$
\[
W_i=(I+M)\cdot\ldots\cdot(I+M^{2^{i-1}})(I+M^{2^{i}})(I-J)(I+M^{2^{i}})(I+M^{2^{i-1}})\cdot\ldots\cdot(I+M)
\]

Since $G$ is connected and aperiodic, as $k$ grows, the term $(I-M^{2^{k}})^{+}$ approaches $(I-J)^{+}$. This is because for all $i,j$ entry $(M^{t})_{ij}$ is the probability that a random walk from $i$ of length $t$ ends at $j$. The graph is regular, so its stationary distribution is uniform and if $t$ is larger than the mixing time of the graph then a random walk of length $t$ starting at $i$ is essentially equally likely to end at any vertex in the graph. The eigenvalues of $I-J$ are all 0 or 1, so $I-J$ is its own pseudoinverse. Undirected multigraphs have polynomial mixing time, so setting $k$ to $O(\log n)$ and then replacing $(I-M^{2^{k+1}})^{+}$ in the final term of the expansion with $(I-J)$ should give a good approximation to $L^{+}$ without explicitly computing any pseudoinverses.

Using the identity directly to approximate $L^{+}$ in $\tilde{O}(\log n)$ space is infeasible because raising a matrix to the $2^{k}$th power by repeated squaring takes space $\Theta(k\cdot\log n)$, which is $\Theta(\log^{2}n)$ for $k=\Theta(\log n)$.

To save space we use the derandomized square introduced in \cite{RozenmanVa05} in place of true matrix squaring. The derandomized square improves the connectivity of a graph almost as much as true squaring does, but it does not square the degree and thereby avoids the space blow-up when iterated. 
\end{section}

\begin{section}{Approximation of the pseudoinverse}
\label{sect:approx}
In this section we show how to approximate the pseudoinverse of a Laplacian using a method from Peng and Spielman \cite{PS13}. The Peng Spielman solver was originally written to approximately solve symmetric diagonally dominant systems, which in turn can be used to approximate solutions to Laplacian systems. We adapt their algorithm to compute the pseudoinverse of a graph Laplacian directly. The approximation proceeds in two steps. The first achieves a constant spectral approximation to the pseudoinverse of a Laplacian matrix. Then the constant approximation is boosted to an $\epsilon$ approximation through rounds of Richardson iteration. 

\begin{theorem}[Adapted from \cite{PS13}]
\label{thm:approx}
Let $\epsilon_{0},\ldots,\epsilon_{k}\geq 0$ and $\epsilon=\sum_{0}^{k}\epsilon_i$. Let $M_0,\ldots,M_{k}$ be symmetric matrices such that $L_{i}=I-M_{i}$ are positive semidefinite for all $i\in\{0,\ldots,k\}$ and $L_{i}\approx_{\epsilon_{i-1}}I-M_{i-1}^2$ for all $1\leq i\leq k$ and $L_{k}\approx_{\epsilon_k}I-J$.
Then 
\[
L_{0}^{+}\approx_{\epsilon}\frac{1}{2}(I-J)+\left(\sum_{i=0}^{k-1}\frac{1}{2^{i+2}}W_{i}\right)+\frac{1}{2^{k+1}}W_{k}
\]
where for all $i\in [k]$
\[
W_i=(I+M_0)\cdot\ldots\cdot(I+M_{i-1})(I+M_{i})(I-J)(I+M_{i})(I+M_{i-1})\cdot\ldots\cdot(I+M_0)
\]
\end{theorem}

The proof of this theorem is adapted from \cite{PS13} and we include it here for completeness. 
\begin{proof}
First we show that for all $0\leq i<k$,
\[
L_{i}^{+}\approx_{\epsilon_{i}}\frac{1}{2}(I-J+(I+M_{i})L_{i+1}^{+}(I+M_{i}))
\]
Let $i<k$. By assumption we have $L_{i+1}\approx_{\epsilon_{i}}I-M_{i}^{2}$. So by Part 7 of Proposition \ref{prop:psdfacts} we have $L_{i+1}^{+}\approx_{\epsilon_{i}}(I-M_{i}^{2})^{+}$. Applying Part 6 of Proposition \ref{prop:psdfacts} gives
\[
(I+M_{i})L_{i+1}^{+}(I+M_{i})\approx_{\epsilon_{i}}(I+M_{i})(I-M_{i}^{2})^{+}(I+M_{i})
\]
Then Parts 2 and 3 of Proposition \ref{prop:psdfacts} tell us that
\begin{align*}
\frac{1}{2}(I-J+(I+M_i)L_{i+1}^{+}(I+M_i))
&\approx_{\epsilon_{i}}\frac{1}{2}(I-J+(I+M_i)(I-M_{i}^{2})^{+}(I+M_i))\\
&=L_{i}^{+}
\end{align*}

Now let 
\[
Z_{i}=\frac{1}{2}\left(I-J+(I+M_{i})Z_{i+1}(I+M_i)\right)
\]
for $i\in\{0,\ldots,k-1\}$ and let $Z_{k}=I-J$. Notice that
\begin{align*}
Z_0&=\frac{1}{2}\left(I-J+(I+M_{0})Z_{1}(I+M_0)\right)\\
&= \frac{1}{2}(I-J)+\frac{1}{4}(I+M_{0})(I-J)(I+M_0)+\frac{1}{4}(I+M_{0})(I+M_1)Z_{2}(I+M_1)(I+M_0)\\
&=\frac{1}{2}(I-J)+\left(\sum_{i=0}^{k-1}\frac{1}{2^{i+2}}W_{i}\right)+\frac{1}{2^{k+1}}W_{k}
\end{align*}
So we want to show that 
\[
Z_0\approx_{\left(\sum_{i=0}^{k}\epsilon_{i}\right)}L_{0}^{+}
\]

We will prove by backwards induction on $i$ that
\[
Z_{i}\approx_{\left(\sum_{j=i}^{k}\epsilon_{j}\right)} L_{i}^{+}
\]
The base case of $Z_{k}\approx_{\epsilon_{k}}(I-J)$ follows by assumption. Supposing the claim holds for $j=i+1$ we show it also holds for $j=i$. By the inductive hypothesis and our assumption about $L_{i+1}$ we have,
\[
Z_{i+1}\approx_{\left(\sum_{j=i+1}^{k}\epsilon_{j}\right)} L_{i+1}^{+}
\]
From Proposition \ref{prop:psdfacts} Part 6, 2, and 3 we have   
\begin{align*}
Z_i&=\frac{1}{2}(I-J(I+M_i)Z_{i+1}(I+M_i))\\
&\approx_{\left(\sum_{j=i+1}^{k}\epsilon_{j}\right)}\frac{1}{2}(I-J+(I+M_i)L_{i+1}^{+}(I+M_i))\\
&\approx_{\epsilon_i}L_i^+
\end{align*}
Applying Proposition \ref{prop:psdfacts} Part 5 then gives
\[
Z_{i}\approx_{\left(\sum_{j=i}^{k}\epsilon_{j}\right)}L_i^+
\]
as desired
\end{proof}

Our algorithm works by using Theorem \ref{thm:approx} to compute a constant approximation to $L^{+}$ and then boosting the approximation to an arbitrary $\epsilon$-approximation. Our main tool for this is the following lemma which shows how an approximate pseudoinverse can be improved to a high quality approximate pseudoinverse. This is essentially a symmetric version of a well known technique in linear system solving known as preconditioned Richardson iteration. It is a slight modification of Lemma~31 from \cite{ls15invmaintenance}.

\begin{lemma}
\label{lem:boostapprox}
Let $A, P \in \R^{n \times n}$ be symmetric invertible positive definite matrices such that for some $\epsilon \in (0, 1)$ it is the case that
$
(1 - \epsilon) P^{-1} \preceq 
A \preceq P^{-1}
$.
Then for all $k > 0$ we have that
$
P_k = \sum_{i = 0}^{k} P (I - A P)^{i} 
$ is a symmetric matrix satisfying
\[
\left(1 - \frac{\epsilon^{k + 1}}{1 - \epsilon}\right)
A^{-1}
\preceq 
P_k \preceq A^{-1} ~.
\]
\end{lemma}
\begin{proof}
Let $M = P^{1/2} A P^{1/2}$ and $M_k = \sum_{i = 0}^{k} (I - M)^i$ so that $P_k = P^{1/2} M_k P^{1/2}$. Note that $(1 - \epsilon) I \preceq M \preceq I$ and therefore $0 \preceq I - M \preceq \epsilon I$. As this implies $\|I - M\|_2 < 1$ we have
\[
P^{-1/2} A^{-1} P^{-1/2} = M^{-1} = (I - (I - M))^{-1}
= \sum_{i = 0}^{\infty} (I - M)^i ~.
\]
Therefore, as $I \preceq M^{-1} = P^{-1/2} A^{-1} P^{-1/2}$ we have
\[
0 \preceq P^{-1/2} A^{-1} P^{-1/2} - P_k
= \sum_{i = k + 1}^{\infty} (I - M)^i
\preceq \left(\sum_{i = k + 1}^{\infty} \epsilon^i\right) I
= \frac{\epsilon^{k + 1}}{1 - \epsilon} I 
= \frac{\epsilon^{k + 1}}{1 - \epsilon} P^{-1/2} A^{-1} P^{-1/2} ~.
\]
This yields that 
\[
\left(1 - \frac{\epsilon^{k + 1}}{1 - \epsilon}\right)
P^{-1/2} A^{-1} P^{-1/2}
\preceq 
 M_k \preceq P^{-1/2} A^{-1} P^{-1/2} ~.
\]
and combining with the fact that $P_k = P^{1/2} M_k P^{1/2}$ yields the result.
\end{proof}
Note that we can apply the same iterative method to boost the approximation quality of a pseudoinverse (rather than true inverse) by carrying out all of the analysis in the space orthogonal to the kernel of the Laplacian.
\begin{corollary}
\label{cor:boostapprox}
Let $\tilde{L}^{+}\approx_{\alpha}L^{+}$ for $\alpha <1/2$. Then
\[
L_k=e^{-\alpha}\cdot\sum_{i = 0}^{k} L\cdot (I - e^{-\alpha}\cdot\tilde{L}^{+}\cdot L)^{i} 
\]
is a matrix satisfying $L_k\approx_{O((2\alpha)^k)} L^{+}$
\end{corollary} 
\begin{proof}
We have that $e^{-\alpha}\tilde{L}^{+}\preceq L^+\preceq e^{\alpha}\tilde{L}^{+}$. Multiplying by $e^{-\alpha}$ gives $e^{-2\alpha}\tilde{L}^{+}\preceq e^{-\alpha}L^+\preceq \tilde{L}^{+}$. Taking $A=e^{-\alpha}L$ and $P=\tilde{L}^{+}$ from Lemma \ref{lem:boostapprox} and $\epsilon=1-e^{-2\alpha}$ we get that for all $k>0$
\[
L_k = \sum_{i = 0}^{k} \tilde{L}^{+}\cdot (I - e^{-\alpha}L\cdot\tilde{L}^{+})^{i} 
\]
satisfies 

\[
\left(1-\frac{(1-e^{-2\alpha})^{k+1}}{e^{-2\alpha}}\right)e^{\alpha}L^+\preceq L_k \preceq e^{\alpha}L^+
\]
Notice that
\begin{align*}
\left(1-\frac{(1-e^{-2\alpha})^{k+1}}{e^{-2\alpha}}\right)&\geq 1-3\cdot(1-e^{-2\alpha})^{k+1}\\
&\geq 1-3\cdot(2\alpha)^{k+1}\\
&\geq e^{-6\cdot(2\alpha)^{k+1}}\\
&=1/e^{O((2\alpha)^{k})}
\end{align*}
which implies
\[
\frac{1}{e^{O((2\alpha)^{k})}}\cdot e^{\alpha}L^+\preceq L_k \preceq e^{\alpha}L^+
\]
Multiplying by $e^{-\alpha}$ completes the proof.  
\end{proof}
\end{section}

\begin{section}{Derandomized Squaring}
To approximate the pseudoinverse of a graph Laplacian we will use Theorem \ref{thm:approx} to get a constant approximation and then use Corollary \ref{cor:boostapprox} to boost the approximation. In order to achieve a space efficient and deterministic implementation of Theorem \ref{thm:approx}, we will replace every instance of matrix squaring with the derandomized square of Rozenman and Vadhan \cite{RozenmanVa05}. Before defining the derandomized square we define two-way labelings and rotation maps.

\begin{definition}{\cite{ReingoldVaWi01}}
A {\em two-way labeling} of a $d$-regular undirected multigraph $G$ is a labeling of the edges in $G$ such that 
\begin{enumerate}
\item Every edge $(u,v)$ has two labels in $[d]$, one as an edge incident to $u$ and one as an edge incident to $v$.
\item For every vertex $v$ the labels of the edges incident to $v$ are distinct.
\end{enumerate}
\end{definition}
In \cite{RozenmanVa05}, two-way labelings are referred to as {\em undirected} two-way labelings. In a two-way labeling, each vertex has its own labeling from $1$ to $d$ for the $d$ edges incident to it. Since every edge is incident to two vertices, each edge receives two labels, which may or may not be the same. It is convenient to specify a multigraph with a two-way labeling by a rotation map:
\begin{definition}{\cite{ReingoldVaWi01}}
Let $G$ be a $d$-regular multigraph on $n$ vertices with a two-way labeling. The {\em rotation map} Rot$_{G}\colon [n]\times[d]\to[n]\times[d]$ is defined as follows: Rot$_{G}(v,i)=(w,j)$ if the $i$th edge incident to vertex $v$ leads to vertex $w$ and this edge is the $j$th edge incident to $w$. 
\end{definition}
Note that Rot$_{G}$ is its own inverse, and that any function Rot$\colon[n]\times[d]\to[n]\times[d]$ that is its own inverse defines a $d$-regular multigraph on $n$ vertices with a two-way labeling. Recall that the edges in $G^2$ correspond to all of the walks of length 2 in $G$. This is equivalent to placing a $d$-clique with self loops on every vertex on the neighbor set of every vertex in $G$. The derandomized square picks out a subset of the walks of length 2 by placing a small degree expander on the neighbor set of every vertex rather than a clique. 

\begin{definition}[\cite{RozenmanVa05}]
\label{def:derandsquare}
Let $G$ be a $d$-regular multigraph on $n$ vertices with a two-way labeling, let $H$ be a $c$-regular multigraph on $d$ vertices with a two-way labeling. The {\em derandomized square} $G\ds H$ is a $(d\cdot c)$-regular graph on $n$ vertices with rotation map Rot$_{G\ds H}$ defined as follows: For $v_{0}\in[n], i_0\in[d]$, and $j_0\in [c]$, we compute Rot$_{G\ds H}(v_0,(i_0,j_0))$ as
\begin{enumerate}
\item Let $(v_1,i_1)=$Rot$_{G}(v_0,i_0)$
\item Let $(i_2,j_1)=$Rot$_{H}(i_1,j_0)$
\item Let $(v_2,i_3)=$Rot$_{G}(v_1,i_2)$
\item Output $(v_2,(i_3,j_1))$
\end{enumerate}
\end{definition}
It can be verified that Rot$_{G\ds H}$ is its own inverse and hence this indeed defines a $(d\cdot c)$-regular multigraph on $n$ vertices. The main idea behind the derandomized square is that it improves the connectivity of the graph (as measured by the second eigenvalue) without squaring the degree.

\begin{theorem}[\cite{RozenmanVa05}]
\label{thm:dsquare}
Let $G$ be a $d$-regular undirected multigraph with a two-way labeling and $\lambda(G)=\lambda$ and $H$ be a $c$-regular graph with a two-way labeling and $\lambda(H)=\mu$. Then $G\ds H$ is a $d\cdot c$-regular undirected multigraph and 
\[
\lambda(G\ds H)\leq 1-(1-\lambda^2)\cdot(1-\mu)\leq \lambda^{2}+\mu
\]
\end{theorem}
The output of the derandomized square is an undirected multigraph with a two-way labeling. So the operation can be repeated on the output. In our algorithm for computing an approximate pseudoinverse, we use  Identity \ref{idexpansion} but replace every instance of a squared transition matrix with the derandomized square. To prove that this approach indeed yields an approximation to the pseudoinverse, we want to apply Theorem \ref{thm:approx}. For this we need to prove two properties of the derandomized square: that the derandomized square is a good spectral approximation of the true square and that iterating the derandomized square yields an approximation of the complete graph. The latter property follows from the corollary and lemma below.
\begin{corollary}
\label{cor:iteratedsquaring}
Let $G_0$ be a $d$-regular undirected multigraph on $n$ vertices with a two-way labeling and $\lambda(G_0)=\lambda$ and $H_1,\ldots,H_k$ be $c$-regular graphs with two-way labelings where $H_i$ has $d\cdot c^{i-1}$ vertices and $\lambda(H_i)\leq \mu$. For all $i\in[k]$ let $G_i=G_{i-1}\ds H_i$. If $\mu\leq 1/100$ and $k=\lceil{6\log d^2 n^2}\rceil$ then $\lambda(G_k)\leq 1/3$.
\end{corollary}
\begin{proof}
Let $\mu\leq 1/100$ and for all $i\in[k]$, let $\gamma_{i}=1-\lambda(G_i)$. If $\gamma_{i-1}< 2/3$ then
\begin{align*}
\gamma_i&\geq (1-\lambda(G_{i-1})^2)\cdot(1-\mu) \\ 
&=(2\gamma_{i-1}-\gamma_{i-1}^2)\cdot(1-\mu) \\ 
&\geq \frac{99}{100}\cdot\frac{4}{3}\gamma_{i-1}\\
&\geq \frac{5}{4}\gamma_{i-1}
\end{align*}
It follows that until $\gamma_{i}$ is driven above $2/3$ we have $\gamma_{k}\geq (5/4)^k\cdot\gamma_0$. Since $\gamma_0\geq 1/2d^2n^2$, setting $k=\lceil{6\log d^2n^2}\rceil$ will result in $\lambda(G_k)\leq 1/3$.
\end{proof}

\begin{lemma}
\label{lem:equivalence}
If $L=I-M$ is the normalized Laplacian of a $d$-regular, undirected, 1/2-lazy, multigraph $G$, and $\epsilon=\ln(1/(1-\lambda))$ then $L\approx_{\epsilon} I-J$ if and only if $\lambda(G)\leq \lambda$.
\end{lemma}
A proof of Lemma \ref{lem:equivalence} can be found in Appendix \ref{app:equivalence}. 

We prove in Theorem \ref{thm:dsapproxessquare} that if $H$ is an expander then the derandomized square of a multigraph $G$ with $H$ is a good spectral approximation to $G^{2}$ and thus can be used for the approximation in Theorem \ref{thm:approx}. The idea is that the square of a graph can be interpreted as putting $d$-clique $K$ on the neighbors of each vertex $v$ in $G$, and the derandomized square can be interpreted as putting a copy of $H$ on each vertex. By Lemma \ref{lem:equivalence}, $L(H)$ spectrally approximates $L(K)$, and we can use this to show that the Laplacian of $L(G \ds H)$ spectrally approximates $L(G^2)$.

\begin{theorem}
\label{thm:dsapproxessquare}
Let $G$ be a $d$-regular, undirected, aperiodic multigraph on $n$ vertices and $H$ be a regular multigraph on $d$ vertices with $\lambda(H)\leq \alpha$. Then 
\[
L(G^{2})\approx_{\epsilon} L(G\ds H)
\]
for $\epsilon=\ln(1/(1-\alpha))$. 
\end{theorem}
\begin{proof}

Following the proof from \cite{RozenmanVa05} that the derandomized square improves connectivity, we can write the transition matrix for the random walk on $G\ds H$ as $M=P \tilde{A} \tilde{B} \tilde{A} Q$, where each matrix corresponds to a step in the definition of the derandomized square:

\begin{itemize}
\item $Q$ is an $n\cdot d\times n$ matrix that maps any $n$-vector $v$ to $v\otimes u_d$ (where $u_d$ denotes the $d$-coordinate uniform distribution). This corresponds to ``lifting'' a probability distribution over $[n]$ to one over $[n]\times[d]$ where the mass on each coordinate is divided uniformly over $d$ coordinates in the distribution over $[n]\times[d]$. That is $Q_{(u,i),v}=1/d$ if $u=v$ and 0 otherwise where the rows of $Q$ are ordered $(1,1),(1,2),\ldots,(1,d),(2,1),\ldots,(2,d),\ldots(n,d)$.
\item $\tilde{A}$ is the $n\cdot d\times n\cdot d$ permutation matrix corresponding to the rotation map for $G$. That is $A_{(u,i),(v,j)}=1$ if Rot$_G(u,i)=(v,j)$ and 0 otherwise.
\item $\tilde{B}$ is $I_n \otimes B$, where $I_n$ is the $n\times n$ identity matrix and $B$ is the transition matrix for $H$.
\item $P$ is the $n\times n\cdot d$ matrix that maps any $(n\cdot d)$-vector to an $n$-vector by summing all the entries corresponding to the same vertex in $G$. This corresponds to projecting distributions on $[n]\times[d]$ back down to a distribution over $[n]$. That is $P_{v,(u,i)}=1$ if $u=v$ and 0 otherwise where the columns of $P$ are ordered $(1,1),(1,2),\ldots,(1,d),(2,1),\ldots,(2,d),\ldots(n,d)$. Note that $P=d\cdot Q^{\intercal}$.
\end{itemize}

Let $c$ be the degree of $H$ and let $K$ be the complete graph on $d$ vertices with self loops on every vertex. Lemma \ref{lem:equivalence} says that $L_H\approx_{\epsilon}L_K$ where $L_H=I-B$ and $L_K=I-J$ are the normalized Laplacians of $H$ and $K$, respectively. It follows that $I_n\otimes L_H = I_{n\cdot d}-\tilde{B}$ $\epsilon$-approximates $I_n\otimes L_K = I_{n\cdot d}-I_n\otimes J$ by Proposition \ref{prop:psdfacts} Part 8. $I_n\otimes L_H$ is the Laplacian for $n$ disjoint copies of $H$ and $I_n\otimes L_K$ is the Laplacian for $n$ disjoint copies of $K$.  Applying the matrices $P\tilde{A}$ and $\tilde{A}Q$ on the left and right place these copies on the neighborhoods of each vertex of $G$.  

Note that $\tilde{A}$ is symmetric because Rot$_{G}$ is an involution. In other words, for all $u,v\in[n]$ and $a,b\in[d]$, Rot$_{G}(u,a)=(v,b) \implies \mathrm{Rot}_{G}(v,b)=(u,a)$. This also implies that $\tilde{A}^{2}=I_{n\cdot d}$. Also note that $PQ=I_n$ and $QP=I_{n\cdot d}\otimes J$. Applying these observations along with Proposition \ref{prop:psdfacts} part 6 we get 
\begin{align*}
I_n-M &= P\tilde{A}I_{n\cdot d}\tilde{A}Q-P\tilde{A}\tilde{B}\tilde{A}Q\\
&=P\tilde{A}(I_n\otimes L_H)\tilde{A}Q \\
& =d\cdot Q^{\intercal}\tilde{A}^{\intercal}(I_n\otimes L_H)\tilde{A}Q\\
&\approx_{\epsilon} d\cdot Q^{\intercal}\tilde{A}^{\intercal}(I_n\otimes L_K)\tilde{A}Q\\
&=P\tilde{A}(I_n\otimes L_K)\tilde{A}Q\\
&=I_n-P\tilde{A}(I_n\otimes J)\tilde{A}Q\\
&=I_n-P\tilde{A}QP\tilde{A}Q\\
&=L(G^2)
\end{align*}
where the final line follows from observing that $P\tilde{A}Q$ equals the transition matrix of $G$. Thus the above shows that $L(G^2)\approx_\epsilon L(G\ds H)$ as desired. 
\end{proof}

Note that if we want to compute a constant spectral approximation to $L^{+}$ then we need each $\epsilon_{i}=1/O(\log n)$ so that the sum of $k=O(\log n)$ of them is a constant. This implies that we need $\alpha=1/O(\log n)$ in Theorem \ref{thm:dsapproxessquare} which requires expanders of polylogarithmic degree. This will affect the space complexity of the algorithm as discussed in the next section.
\end{section}

\begin{section}{Space complexity}
\label{sect:space}
In this section we argue that given a Laplacian $L$, our algorithm produces an $\epsilon$-approximation to $L^+$ using $O\left(\log n\cdot\log\log(n/\epsilon)\right)$ space. First we prove that when $L$ is the normalized Laplacian of a regular graph, we can produce a constant-approximation to $L^{+}$ using $O\left(\log n\cdot\log\log n\right)$ space and then boost the approximation quality to arbitrary $\epsilon$ accuracy using an additional $O(\log n\cdot\log\log (1/\epsilon))$ space thereby proving Theorem \ref{thm:arbitraryapproxmain}.

\begin{subsection}{Model of space bounded computation}
\label{sect:spacemodel}
We use a standard model of space bounded computation described here. The machine has a read-only input tape, a constant number of  read/write work tapes, and a write-only output tape. We say the machine runs in space $s$ if throughout the computation, it only uses $s$ tape cells on the work tapes. The machine may write outputs to the output tape that are larger than $s$ (in fact as large as $2^{O(s)}$) but the output tape is write-only. For the randomized analogue, we imagine the machine endowed with a coin-flipping box. At any point in the computation the machine may ask for a random bit but in order to remember any bit it receives, the bit must be recorded on a work tape. 

We use the following proposition about the composition of space-bounded algorithms.
\begin{proposition}
\label{prop:composition}
Let $f_1$ and $f_2$ be functions that can be computed in space $s_1(n),s_2(n)\geq \log n$, respectively, and $f_2$ has output of length $\ell_1(n)$ on inputs of size $n$. Then $f_2\circ f_1$ can be computed in space 
\[
O(s_2(\ell_1(n)) + s_1(n)).
\]
\end{proposition}
\end{subsection}

\begin{subsection}{Constant approximation of the pseudoinverse}
In this section we show how to compute a constant approximation of the pseudoinverse of the normalized Laplacian of a regular multigraph in space $O(\log n\cdot\log\log n)$ where $n$ is the bitlength of the input. 
\begin{proposition}
\label{prop:constantapprox}
There is an algorithm that given an undirected, aperiodic, $d$-regular multigraph $G$ with normalized Laplacian $L$, and $d$ a power of 2, computes a matrix $\tilde{L}^{+}$ such that $\tilde{L}^{+}\approx_{\alpha}L^{+}$ using space $O\left(\log n\cdot\log\log n\right)$ where $\alpha < 1/2$ and $n$ is the bitlength of the input.
\end{proposition}
To prove Proposition \ref{prop:constantapprox}, we first prove a few lemmas. We will use the fact that large derandomized powers can be computed in small space, which we prove following \cite{Reingold08,RozenmanVa05,vadhan2012pseudorandomness}. First we note that neighbors in the sequence of expanders we use for the iterated derandomized square can be explicitly computed space-efficiently.

\begin{lemma}
\label{lem:expanders}
For every $t\in \mathbb{N}$ and $\mu>0$, there is a graph $H_{t,\mu}$ with a two-way labeling such that:
\begin{itemize}
\item $H$ has $2^t$ vertices and is $c$-regular for $c$ being a power of 2 bounded by $\poly(t,1/\mu)$.
\item $\lambda(H) \leq \mu$
\item \textup{Rot}$_H$ can be evaluated in linear space in its input length, i.e. space $O(t+\log c)$.
\end{itemize}
\end{lemma}
\begin{proof}[Proof sketch.]
The expanders can be obtained by letting $H_{t,\lambda}$ be the Cayley graph on vertex set $\mathbb{F}_{2}^t$ with neighbors specified by a $\lambda$-biased multiset $S\subseteq \mathbb{F}_2^{t}$ \cite{NaorNa93,AlonGoHaPe92}, taking Rot$_H(v,i)=(v+s_i,i)$, where $s_i$ is the $i$th element of $S$.
\end{proof}

Now we argue that high derandomized powers can be computed space-efficiently. 
\begin{lemma}
\label{lem:derandspacecomplexity}
Let $G_{0}$ be a $d$-regular, undirected multigraph on $n$ vertices with a two-way labeling and $H_1,\ldots,H_k$ be $c$-regular undirected graphs with two-way labelings where for each $i\in[k]$, $H_{i}$ has $d\cdot c^{i-1}$ vertices. For each $i\in[k]$ let 
\[
G_i=G_{i-1}\ds H_i.
\]
Then given $v_0\in[n], i_{0}\in[d\cdot c^{i-1}], j_0\in[c]$, \textup{Rot}$_{G_i}(v,(i_0,j_0))$ can be computed in space $O(\log(n\cdot d)+k\cdot\log c)$ with oracle queries to \textup{Rot}$_{H_1},\ldots,\mathrm{Rot}_{H_k}$.
\end{lemma}
\newcommand{\Space}{\mathrm{Space}}
\begin{proof}
When $c$ is subpolynomial we are reasoning about sublogarithmic space complexity, which can depend on the model. So we will be explicit about the model we are using. We compute the rotation map of $G_i$ on a multi-tape Turing machine with the following input/output conventions:

\begin{itemize}
\item Input Description:
\begin{itemize}
\item Tape 1 (read-only): Contains the initial input graph $G_0$, with the head at the leftmost position of the tape.
\item Tape 2 (read-write): Contains the input triple $(v_0,(i_0,j_0))$, where $v_0$ is a vertex of $G_i$, $i_0\in [d\cdot c^{i-1}]$ is an edge label in $G_i$, and $j_0\in[c]$ is an edge label in $H_i$ on a {\em read-write} tape, with the head at the {\em rightmost} position of $j_{0}$. The rest of the tape may contain additional data.
\item Tapes 3+ (read-write): Blank worktapes with the head at the leftmost position.
\end{itemize}

\item Output Description:
\begin{itemize}
\item Tape 1: The head  should be returned to the leftmost position.
\item Tape 2: In place of $(v_0,(i_0,j_0))$, it should contain the output $(v_2,(i_3,j_1))=\mathrm{Rot}_{G_i}(v_0,(i_0,j_0))$ as described in Definition \ref{def:derandsquare}. The head should be at the rightmost position of $j_1$ and the rest of the tape should remain unchanged from its state at the beginning of the computation.
\item Tapes 3+ (read-write): Are returned to the blank state with the heads at the leftmost position.
\end{itemize}
\end{itemize}

Let $\Space(G_i)$ be the amount of space required to compute the rotation map of graph $G_i$. We will show that for all $i\in[k]$, $\Space(G_i)=\Space(G_{i-1})+O(\log c)$. Note that $\Space(G_0)=O(\log (nd))$. 

Fix $i\in [k]$. We begin with $v_0\in[n], i_{0}\in[d\cdot c^{i-1}]$ and $j_0\in[c]$ on tape 2 with the head on the rightmost position of $j_0$ and we want to compute Rot$_{G_i}(v_0,(i_0,j_0))$. We move the head left to the rightmost position of $i_0$, recursively compute Rot$_{G_{i-1}}(v_0,i_0)$ so that tape 2 now contains $(v_1,i_1,j_0)$.  Then we move the head to the rightmost position of $j_0$ and compute Rot$_{H_i}(i_1,j_0)$ so that tape 2 now contains $(v_1,i_2,j_1)$. Finally, we move the head to the rightmost position of $i_2$ and compute Rot$_{G_{i-1}}(v_1,i_2)$ so that tape 2 contains $(v_2,i_3,j_1)$. 

This requires 2 evaluations of the rotation map of $G_{i-1}$ and one evaluation of the rotation map of $H_i$. Note that we can reuse the same space for each of these evaluations because they happen in succession. The space needed on top of $\Space(G_{i-1})$ is the space to store edge label $j_0$, which uses $O(\log c)$ space. So $\Space(G_i)=\Space(G_{i-1})+O(\log c)$. Since $i$ can be as large as $k$ and $\Space(G_0)=O(\log (n\cdot d))$ we get that for all $i\in[k]$, $\Space(G_i)=O(\log n\cdot d + k\cdot\log c)$.
\end{proof}

\begin{corollary}
\label{cor:entriesofM}
Let $M_0,\ldots,M_k$ be the transition matrices of $G_0,\ldots, G_k$ as defined  in Lemma \ref{lem:derandspacecomplexity}. For all $\ell\in[k]$, given coordinates $i,j\in [n]$, entry $i,j$ of $M_\ell$ can be computed in space $O(\log n\cdot d + k\cdot\log c)$.
\end{corollary}
\begin{proof}
Lemma \ref{lem:derandspacecomplexity} shows that we can compute neighbors in the graph $G_\ell$ in space $O(\log n\cdot d + k\cdot\log c)$. Given coordinates $i,j$ the algorithm initiates a tally $t$ at 0 and computes Rot$_{G_\ell}(i,q)$ for each $q$ from $1$ to $d\cdot c^{\ell-1}$, the degree of $G_\ell$. If the vertex outputted by Rot$_{G_\ell}$ is $j$, then $t$ is incremented by 1. After the loop finishes, $t$ contains the number of edges from $i$ to $j$ and the algorithm outputs $t/d\cdot c^{\ell-1}$, which is entry $i,j$ of $M_\ell$. This used space $O(\log n\cdot d + k\cdot\log c)$ to compute the rotation map of $G_\ell$ plus space $O(\log (d\cdot c^{\ell-1}))$ to store $q$ and $t$. So the total space usage is $O(\log n\cdot d + k\cdot\log c) +O(\log (d\cdot c^{\ell-1}))=O(\log n\cdot d + k\cdot\log c)$.
\end{proof}
It follows from Corollary \ref{cor:entriesofM} that when $k=O(\log n)$ for all $i\in[k]$ entries in $I+M_i$ can be computed in space $O(\log n\cdot\log c)$. When we apply it, we will use expanders (from Lemma \ref{lem:expanders}) of degree $c=O(\polylog (n))$, so our algorithm computes entries of $I+M_i$ in space $O(\log n\cdot\log \log n)$.

Now we prove that we can multiply matrices in small space. 

\begin{lemma}
\label{lem:matrixprod}
Given matrices $M^{(1)},\ldots,M^{(k)}$ and indices $i,j$ $(M^{(1)}\cdot M^{(2)}\cdot \ldots\cdot M^{(k)})_{ij}$ can be computed using $O(\log n\cdot \log k)$ space where $n$ is the dimension of the input matrices.
\end{lemma}

\begin{proof}
First we show how to multiply two $n\times n$ matrices in $O(\log n)$ space. Given as input matrices $M^{(1)},M^{(2)}$ and indices $i,j$ we wish to compute 
\[
(M^{(1)}\cdot M^{(2)})_{ij}=\sum_{\ell=1}^{n}M^{(1)}_{i\ell}M^{(2)}_{\ell j}
\]
We use the fact that $n$ $n$-bit numbers can be multiplied and added in $O(\log n)$ space (in fact in TC$^{0}$ \cite{BCH86,ABH02}). So for a counter $\ell$ from $1$ to $n$, the algorithm multiplies $M_{i\ell}\cdot M_{\ell j}$ and adds the result. The counter can be stored with $\log n$ bits and the arithmetic can be carried out in logspace.

To multiply $k$ matrices we recursively split the product into two blocks and multiply each block separately. The pseudocode for multiplying $k$ matrices is below

\begin{algorithm}[H]
\begin{algorithmic}
\Function{mult}{$M^{(1)},M^{(2)},\ldots,M^{(k)}, i,j$}
\If {$k=1$}
    \State \Return $M^{(1)}_{ij}$
\Else
\If {$k=2$}
    \State \Return $(M^{(1)}\cdot M^{(2)})_{ij}$
\Else
\State \Return $\sum_{\ell=1}^{n}$\Call{mult}{$M^{(1)}, \ldots, M^{(\lfloor{k/2\rfloor})}, i,\ell$}$\cdot$ \Call{mult}{$M^{({\lfloor{k/2\rfloor}}+1)}, \ldots, M^{(k)},\ell,j$}
\EndIf
\EndIf
\EndFunction
\end{algorithmic}
\end{algorithm}
The depth of the recursion is $\log k$ and each level requires $O(\log n)$ space for a total of $O(\log n\cdot \log k)$ space.
\end{proof}  
Now we can prove Proposition \ref{prop:constantapprox}.
\begin{proof}[Proof of Proposition \ref{prop:constantapprox}]
Let $M_0$ be the transition matrix of $G_0=G$ such that $L=I-M_0$.  Clearly a two-way labeling of $G$ can be computed in $O(\log n)$ space by just fixing a canonical way for each vertex to label its incident edges. Set $k=\lceil{6\log d^2 n^2}\rceil$ and $\mu=1/30k$. Let $c=\polylog(n)$ be a power of 2 and let $t_i=\log d\cdot c^{i-1}$ for all $i\in[k]$. Note that each $t_i$ is an integer because $d$ and $c$ are powers of 2. Let $H_1,\ldots,H_k$ be $c_i$-regular graphs on $2^{t_i}=d\cdot c^{i-1}$ vertices, respectively and $\lambda(H_i)\leq \mu$ as given by Lemma \ref{lem:expanders}. So for all $i\in[k]$, $c_i=\poly(t_i,1/\mu)=\polylog(n)$. Without loss of generality, we can take $c=\polylog(n)\geq c_i$ for all $i\in[k]$ and make each expander $c$-regular by adding self loops. 

For all $i\in[k]$ let $G_i=G_{i-1}\ds H_i$ and let $M_i$ be the transition matrix of $G_i$. By Theorem \ref{thm:dsapproxessquare} we have $I-M_i\approx_{\epsilon}I-M_{i-1}^2$ for all $i\in [k]$ where $\epsilon=\ln(1/(1-\mu))$. By Corollary \ref{cor:iteratedsquaring} $\lambda(G_{k})\leq 1/3$ and so by Lemma \ref{lem:equivalence} we have $I-M_{k}\approx_{\ln(3/2)} I-J$.

From Theorem \ref{thm:approx}, we get that 
\begin{equation}
\label{eq:constantapprox}
L^{+}\approx_{\delta}\frac{1}{2}(I-J)+\left(\sum_{i=0}^{k-1}\frac{1}{2^{i+2}}W_{i}\right)+\frac{1}{2^{k+1}}W_{k}
\end{equation}
where for all $i\in [k]$
\[
W_i=(I+M_0)\cdot\ldots\cdot(I+M_{i-1})(I+M_{i})(I-J)(I+M_{i})(I+M_{i-1})\cdot\ldots\cdot(I+M_0)
\]
for $\delta=k\cdot\ln(1/(1-\mu))+\ln(3/2)< .5$ (using the fact that $e^{-2\mu}\leq 1-\mu$). So we need to compute approximation \ref{eq:constantapprox} above in $O(\log n\cdot\log\log n)$ space. There are $k+2=O(\log n)$ terms in the sum. Aside from the first term (which is easy to compute) and adjusting the coefficient on the last term, the $i$th term in the expansion looks like  
\[
\frac{1}{2^{i+2}}(I+M_{0})\ldots\cdot(I+M_{i})(I-J)(I+M_{i})\cdot\ldots\cdot(I+M_{0})\\
\]
which is the product of at most $O(\log n)$ matrices. Corollary \ref{cor:entriesofM} says that for all $i\in[k]$ we can compute the entries of $M_i$ in space $O(\log n\cdot d+k\log c)=O(\log n\cdot\log\log n)$. Lemma \ref{lem:matrixprod} says that we can compute the product of $O(\log n)$ matrices in $O(\log n\cdot\log\log n)$ space. By the composition of space bounded algorithms, each term in the sum can be computed in space $O(\log n\cdot\log\log n)+O(\log n\cdot\log\log n)=O(\log n\cdot\log\log n)$. Then since iterated addition can be carried out in $O(\log n)$ space \cite{BCH86,ABH02}, the terms of the sum can be added using an additional $O(\log n)$ space. Again by the composition of space bounded algorithms, the total space usage is $O(\log n\cdot\log\log n)$ for computing a constant approximation to $L^{+}$.
\end{proof}
\end{subsection}
\begin{subsection}{A more precise pseudoinverse}
Now that we have computed a constant approximation to $L^{+}$, we show how to improve the quality of our approximation through iterative methods. 
\begin{proposition}
\label{prop:boostapprox}
Let $L$ be the normalized Laplacian of an undirected, regular, aperiodic, multigraph. There is an algorithm such that for every constant $\alpha<1/2$, given $\tilde{L_1}$ such that $\tilde{L_1}\approx_{\alpha} L^{+}$, computes $\tilde{L_2}$ such that $\tilde{L_2}\approx_{\epsilon}L^{+}$ using space $O(\log n\cdot\log\log n + \log n\cdot\log\log(1/\epsilon))$ where $n$ is the length of the input. 
\end{proposition}

\begin{proof}
By Corollary \ref{cor:boostapprox}, given $\tilde{L_1}\approx_{\alpha}L$, we can boost the approximation from $\alpha$ to $\epsilon$ by setting $k=O(\log (1/\epsilon))$ and computing 
\[
\tilde{L_2}=e^{-\alpha}\cdot\sum_{i=0}^{k}L\cdot(I-e^{-\alpha}\tilde{L_1}\cdot L)^{i}
\]
By Proposition \ref{prop:constantapprox}, entries in $\tilde{L_1}$ can be computed in space $O(\log n\cdot\log\log n)$ and hence entries in $I-e^{-\alpha}\cdot\tilde{L_1}\cdot L$ can be computed in space $O(\log n\cdot\log\log n)$ because matrix addition and multiplication can be done in $O(\log n)$ space. Viewing $(I-e^{-\alpha}\cdot\tilde{L_1}\cdot L)$ as a single matrix, each term in the sum above is the product of at most $i+1$ matrices. Lemma \ref{lem:matrixprod} tells us that we can compute such a product in space $O(\log n \cdot \log i)$. Since $i$ can be as large as $k=O(\log 1/\epsilon)$, this gives space $O(\log n\cdot\log\log(1/\epsilon))$ for computing each term. Then since iterated addition can be computed in logarithmic space \cite{BCH86, ABH02}, the $k+1$ terms in the sum can be added using an additional $O(\log N)$ space where $N$ is the bit length of the computed terms. By composition of space bounded algorithms, computing $\tilde{L_2}$ uses a total of $O(\log n\cdot\log\log n+\log n\cdot\log\log(1/\epsilon))$ space.
\end{proof}
\end{subsection}
\end{section}
\begin{section}{Proof of main result}
Now we are ready to prove our main technical result, Theorem \ref{thm:arbitraryapproxmain}. Theorem \ref{thm:main} follows as a corollary and is discussed in Section \ref{sec:applications}. 
\begin{proof}[Proof of Theorem \ref{thm:arbitraryapproxmain}]
Let $L=D-A$ be the input Laplacian of the graph $G$. Let $\Delta$ be the maximum vertex degree in $G$. We construct a graph $G'$ as follows: at every vertex $v$, we add $2^{\lceil{\log 2\Delta}\rceil}-d(v)$ self loops where $d(v)$ denotes the degree of $v$. Note that $2^{\lceil{\log 2\Delta}\rceil}-d(v)\geq \Delta$ so this ensures $G'$ is aperiodic and $2^{\lceil{\log 2\Delta}\rceil}$-regular. Let $E$ be the diagonal matrix of self loops added in this stage. Then the degree matrix of $G'$ is $D+E=2^{\lceil{\log 2\Delta}\rceil}\cdot I$  and the adjacency matrix of $G'$ is $A+E$. So the unnormalized Laplacian of $G'$ is $D+E-A-E=D-A$, which is the same as the unnormalized Laplacian of $G$. Let $L'=(D-A)(D+E)^{-1}=I-M$ be the normalized Laplacian of $G'$ where $M$ is the transition matrix of $G'$. 

By Proposition \ref{prop:constantapprox}, we can compute a matrix $\tilde{L_1}$ such that $\tilde{L_1}\approx_{\alpha}L'^{+}$ for $\alpha<1/2$ in space $O(\log n\cdot\log\log n)$. Then we can compute $\tilde{L_2}\approx_{\epsilon} L'^+$ by applying Proposition \ref{prop:boostapprox} using an additional $O(\log n\cdot\log\log n+\log n\cdot \log\log (1/\epsilon))$ space. By composition of space bounded algorithms this yields a $O(\log n\cdot\log\log n+\log n\cdot \log\log (1/\epsilon))=O(\log n\cdot\log\log (n/\epsilon))$ space algorithm for computing an $\epsilon$-approximation to $L'^{+}$. Recalling that $L'=L/2^{\lceil{\log 2\Delta}\rceil}$ implies that $\tilde{L_2}/2^{\lceil{\log 2\Delta}\rceil}\approx_{\epsilon} L^{+}$ as desired.

\end{proof}
\end{section}
\begin{section}{Corollaries} \label{sec:applications}
In this section we prove some applications of Theorem \ref{thm:arbitraryapproxmain}. 
\begin{definition}
Let $L$ be the normalized Laplacian of an undirected multigraph and $b\in\mathrm{im}(L)$. Then $\tilde{x}$ is an {\em $\epsilon$-approximate solution} to the system $Lx=b$ if there is an actual solution $x^{*}$ such that
\[
\|x^{*}-\tilde{x}\|_{L} \leq \epsilon\|x^{*}\|_{L}
\]
where for all $v\in\mathbb{R}^{n}$, $\|v\|_{L}\equiv\sqrt{v^{\intercal}Lv}$.
\end{definition} 

Here we prove that our algorithm for computing an $\epsilon$-approximation to the pseudoinverse of a Laplacian can be used to solve Laplacian systems (Theorem \ref{thm:main}). The following lemma is useful for translating between the $L$-norm and the $\ell_2$ norm. 
\begin{lemma}
\label{lem:normtranslation}
Let $x\in\mathbb{R}^n$ such that $x\perp\vec{1}$ and $L$ be the normalized Laplacian of an undirected  multigraph with smallest nonzero eigenvalue $\gamma_2$ and largest eigenvalue $\gamma_n$. Then
\[
\gamma_2\|x\|_2^2 \leq \|x\|_L^2\leq \gamma_n\|x\|_2^2
\]
\end{lemma}
\begin{proof}
By definition $\|x\|_L^2 = x^{\intercal}Lx$ and by the variational characterization of the eigenvalues it follows that
\[
\gamma_2\cdot\|x\|_2^2=\gamma_2\cdot x^{\intercal}x\leq \|x\|_L^2\leq \gamma_n\cdot x^{\intercal}x=\gamma_n\cdot\|x\|_2^2
\]
\end{proof}
Since for undirected multigraphs with maximum degree $d$, $\gamma_2\geq 1/2dn^2$ and $\gamma_n\leq 2d$ it follows from the above lemma that for all $x\perp \vec{1}$, $\|x\|_2$ and $\|x\|_L$ are within a multiplicative factor of $2d\cdot n$ of one another. We also use the following fact.
\begin{lemma} 
\label{lem:pseudoinv_implies_sol}
If $\tilde{L}^{+}\approx_{\epsilon}L^{+}$, $\epsilon\leq \ln(2)$ and $b\in \mathrm{im}(L)$ then $\tilde{x}=\tilde{L}^{+}b$ is a $\sqrt{2\epsilon}$-approximate solution to $Lx=b$.
\end{lemma} 
For completeness, we include a proof of Lemma \ref{lem:pseudoinv_implies_sol} in Appendix \ref{app:pseudoinv_implies_sol}. Now using Theorem \ref{thm:arbitraryapproxmain} and the lemmas above we can prove Theorem \ref{thm:main}.

\begin{proof}[Proof of {\em Theorem \ref{thm:main}}]
Let $L=D-A$ be the input Laplacian of the graph $G$ and let $b\perp \vec{1}$ be the given vector. If $G$ is disconnected then we can use Reingold's algorithm to find the connected components and work on each component separately. So assume without loss of generality that $G$ is connected. We may also assume without loss of generality that G is $d$-regular. Let $\epsilon'=(\epsilon/(4d^2n^2))^2/2$ where $n$ is the bit length of the input. Theorem \ref{thm:arbitraryapproxmain} says that we can compute $\tilde{L}^{+}\approx_{\epsilon'} L^{+}$ in space  $O(\log n\cdot \log\log(n/\epsilon'))=O(\log n\cdot \log\log(n/\epsilon))$.

By Lemma \ref{lem:pseudoinv_implies_sol}, $x=\tilde{L}^{+}b$ is a $\sqrt{2\epsilon'}$-approximate solution to $Lx=b$. In other words, letting $x^{*}=L^{+}b$ we have
\[
\|x-x^{*}\|_{L}\leq \sqrt{2\epsilon'}\cdot\|x^{*}\|_{L}
\]
Translating to the $\ell_2$ norm we get
\[
\frac{1}{2dn}\cdot \|x-x^{*}\|_2\leq \|x-x^{*}\|_{L}\leq \sqrt{2\epsilon'}\cdot\|x^{*}\|_{L}\leq 2dn\sqrt{2\epsilon'}\cdot\|x^{*}\|_2
\]
which implies
\[
\|x-x^{*}\|_2\leq (2dn)^2\sqrt{2\epsilon'}\cdot\|x^{*}\|_2=\epsilon\cdot\|x^{*}\|_2
\]
as desired. 
\end{proof}

For some applications, including the ones listed below, it is useful to be able to apply the pseudoinverse of the Laplacian normalized by the degrees $((D-A)D^{-1})^{+}$ to a vector rather than apply $(D-A)^{+}$. Our algorithm can be used for this task as well. 
\begin{corollary}
\label{cor:normalized}
There is a deterministic algorithm that takes as input the Laplacian of an undirected, multigraph $D-A$, a vector $b\perp\vec{1}$, and $\epsilon>0$ and computes a vector $\tilde{x}$ such that 
\[
\|\tilde{x}-((D-A)D^{-1})^{+}b\|_2 \leq \epsilon\cdot\|x^{*}\|_2
\]
where $x^{*}$ is some solution to the system $(D-A)D^{-1}x=b$ and uses space $O(\log n\cdot\log\log (n/\epsilon))$ where $n$ is the bitlength of the input. 
\end{corollary}
\begin{proof}
We may assume without loss of generality that $G$ is connected. Let $u=D(D-A)^{+}b$ and $v=((D-A)D^{-1})^{+}b$. Let $d=\text{diag}(D)$ and $d'$ denote the vector with $1/d_i$ in each coordinate $i$. Note that $u\perp d'$ and $v\perp d$. Both $u$ and $v$ are solutions to the system $(D-A)D^{-1}x=b$ so there is some scalar $c$ such that $u-v=c\cdot d$ because the multiples of $d$ form the kernel of $(D-A)D^{-1}$. Next we find the scalar $c$ by observing:
\begin{align*}
d^{\intercal}u &= d^{\intercal}(u-v) \\ 
&=c\cdot d^{\intercal}d\\
&=c\cdot\|d\|_2^2\\
\end{align*}
which implies $c= d^{\intercal}u /\|d\|_2^2$. So we have 
\begin{align*}
v &= u-\left(d^{\intercal}u /\|d\|_2^2\right)\cdot d \\
&=\left(D-\frac{dd^{\intercal}}{\|d\|_2^2}\cdot D\right)\cdot (D-A)^{+}b 
\end{align*}
Let $\epsilon'=\epsilon/\sqrt{|V|}\cdot\|d\|_2$ where $|V|$ is the number of vertices in the input graph. By applying Theorem \ref{thm:main}, we can compute a vector $\tilde{z}$ such that 
\[
\|\tilde{z}-(D-A)^{+}b\|_2\leq \epsilon'\cdot\|(D-A)^{+}b\|_2
\]
in $O(\log n\cdot\log\log (n/\epsilon'))=O(\log n\cdot\log\log (n/\epsilon))$ space. Let $K=\left(D-\frac{dd^{\intercal}}{\|d\|_2^2}\cdot D\right)$ and let $\tilde{x}=K\cdot \tilde{z}$ and observe:
\begin{align*}
\|\tilde{x}-v\|_2 &= \left\|K(\tilde{z}-(D-A)^{+}b)\right\|_2 \\
&\leq \epsilon'\cdot\|(D-A)^{+}b\|_2\cdot\sqrt{\sum_{i,j}\left|K_{ij}\right|^2}
\end{align*}
Note that $|K_{ij}|\leq d_i$ for all $i,j$ and so \[\sqrt{\sum_{i,j}\left|K_{ij}\right|^2}\leq \sqrt{|V|}\cdot\|d\|_2\]
Plugging this in above gives
\begin{align*}
\|\tilde{x}-v\|_2 &\leq \epsilon\cdot\|(D-A)^{+}b\|_2\\
&\leq \epsilon\cdot\|u\|_2
\end{align*}
as desired. The matrix $K$ can also be computed in $O(\log n)$ space and so $\tilde{x}=K\cdot\tilde{z}$ can be computed in $O(\log n\cdot\log\log (n/\epsilon))$ space.
\end{proof}

Our algorithm also implies a $\tilde{O}(\log n)$ space algorithm for computing many interesting quantities associated with undirected graphs. Below we survey a few of these corollaries including hitting times, commute times, and escape probabilities. 
\begin{definition}
In a multigraph $G=(V,E)$ the {\em hitting time} $H_{uv}$ from $u$ to $v\in V$ is the expected number of steps a random walk starting at $u$ will take before hitting $v$. 
\end{definition}

\begin{definition}
In a multigraph $G=(V,E)$ the {\em commute time} $C_{uv}=H_{uv}+H_{vu}$ between vertices $u,v\in V$ is the expected number of steps a random walk starting at $u$ will take to reach $v$ and then return to $u$. 
\end{definition}

\begin{corollary}
\label{cor:hittingcommute}
Given an undirected multigraph $G=(V,E)$, vertices $u,v\in V$ and $\epsilon>0$ there are deterministic $O(\log n\cdot\log\log(n/\epsilon))$ space algorithms for computing numbers $\tilde{H}_{uv}$ and $\tilde{C}_{uv}$ such that 
\[
|\tilde{H}_{uv}-H_{uv}|\leq \epsilon 
\]
and 
\[
|\tilde{C}_{uv}-C_{uv}|\leq \epsilon 
\]
where $n$ is the bitlength of the input.
\end{corollary}
\begin{proof}
Let $e_u, e_v$ be the vectors with a 1 in coordinates $u$ and $v$, respectively and zeros elsewhere. Let $M$ be the transition matrix of the random walk on $G$ and let $d_v$ be the degree of vertex $v$ in $G$. It is shown in \cite{CKPPSV16} that 
\[
H_{uv}=(\vec{1}-e_v/d_v)^{\intercal}(I-M)^{+}(e_u-e_v)
\]
We can compute an approximation to the vector $(I-M)^{+}(e_u-e_v)$ in space $O(\log n\cdot\log\log (n/\epsilon))$ using Corollary \ref{cor:normalized} and multiply the result by $(\vec{1}-e_v/d_v)^{\intercal}$ using an additional $O(\log n)$ space. To approximate $C_{uv}$, we simply approximate $H_{uv}$ and $H_{vu}$ and sum them.

Let $x^{*}=(I-M)^{+}(e_u-e_v)$ and $y=\vec{1}-e_v/d_v$. Let $\tilde{x}$ be the approximation of $x^{*}$ we compute using Corollary \ref{cor:normalized}.

For the approximation quality, note that
\begin{align*}
|y\cdot \tilde{x}-y\cdot x^{*}| &\leq \|y\|\cdot\|\tilde{x}-x^{*}\|\\
&\leq \sqrt{n}\cdot O(d^2n^2)\cdot\epsilon 
\end{align*}
where the first line uses Cauchy-Schwarz and the second line follows by Lemma \ref{lem:normtranslation} and Lemma \ref{lem:pseudoinv_implies_sol}. 
Changing $\epsilon$ by a factor of $\poly(n,d)$ completes the proof. 
\end{proof}
Note that since the values we are approximating are at least 1 (if not 0), Corollary \ref{cor:hittingcommute} also implies that we can achieve relative $\epsilon$ error. Our algorithm can also be used to compute escape probabilities.

\begin{definition}
In a graph $G=(V,E)$, for two vertices $u$ and $v$, the {\em escape probability} $p_{w}(u,v)$ denotes the probability that a random walk starting at vertex $w$ reaches $u$ before first reaching $v$. 
\end{definition}
\begin{corollary}
Given an undirected multigraph $G=(V,E)$, vertices $u,v\in V$, and $\epsilon>0$, there is a deterministic $O(\log n\cdot\log\log (n/\epsilon))$ space algorithm for computing a vector $\tilde{p}$ such that 
\[
\|\tilde{p}-p\|\leq \epsilon 
\]
where $p$ is the vector of escape probabilities $p_i(u,v)$ in $G$.
\end{corollary}
\begin{proof}
Let $M$ be the transition matrix of the random walk on $G$, $s$ be the stationary distribution of $G$ and $S$ be the diagonal matrix with the stationary distribution on the diagonal. For undirected graphs $s_u=d(u)/2|E|$ where $d(u)$ is the degree of vertex $u$ in $G$, so $s$ and $S$ can be computed in $O(\log n)$ space. Let $e_u$, $e_v$ be the vectors with a 1 in entries $u,v$, respectively and zeros elsewhere. Let $p$ be the vector of escape probabilities for vertices $u$ and $v$. It is shown in \cite{CKPPSV16} that 
\[
p=\beta(\alpha\cdot s+(I-M)^{+}(e_u-e_v))
\]
where 
\[
\alpha =-e_v^{\intercal}S^{-1}(I-M)^{+}(e_u-e_v)
\]
and 
\[
\beta = \frac{1}{s_u(e_u-e_v)^{\intercal}S^{-1}(I-M)^{+}(e_u-e_v)}
\]
So computing an approximation of $p$ amounts to approximating $(I-M)^{+}(e_u-e_v)$, which we can do using Corollary \ref{cor:normalized}, along with a constant number of multiplications and additions of vectors, which can all be computed in $O(\log n\cdot\log\log (n/\epsilon))$ space.

Let $x^{*}=(I-M)^{+}(e_u-e_v)$ and $y=-e_v^{\intercal}S^{-1}$. Let $\tilde{x}$ be the approximation of $x^{*}$ we compute using Corollary \ref{cor:normalized}. 

\begin{align*}
\|y\cdot \tilde{x}- y\cdot x^{*}\|&\leq \|y\|\cdot\|\tilde{x}-x^{*}\|\\
&\leq \sqrt{n}\cdot O(d^2n^2)\cdot\epsilon 
\end{align*}
Changing $\epsilon$ by a factor of $\poly(n,d)$ shows that we can get a $1/\poly(n,d)$ approximation to $\alpha$. A similar argument shows that we can achieve a $1/\poly(n,d)$ approximation of $\beta$ (noting that $\beta\leq \poly(n,d)$ as shown in \cite{CKPPSV16}). The approximation of the escape probability vector $p$ follows.
\end{proof}
\end{section}
\newpage

\bibliographystyle{alphanum}
\bibliography{lap,pseudorandomness}

\newpage 
\appendix
\begin{section}{Proof of Proposition \ref{prop:psdfacts}}
\label{app:psdfacts}
Fix $v\in\mathbb{R}^{n}$ and $\epsilon, \epsilon_1, \epsilon_2\geq 0$. 
\begin{enumerate}
	\item Suppose $e^{\epsilon}v^{\intercal}Xv\geq v^{\intercal}Yv\geq e^{-\epsilon}v^{\intercal}Xv$. Multiplying the whole expression by $e^{-\epsilon}$ gives 
    \[
   v^{\intercal}Xv\geq e^{-\epsilon}v^{\intercal}Yv\geq e^{-2\epsilon}v^{\intercal}Xv
    \]
    Multiplying by $e^{\epsilon}$ gives
    \[
   e^{2\epsilon}v^{\intercal}Xv\geq e^{\epsilon}v^{\intercal}Yv\geq v^{\intercal}Xv
    \]
    It follows that
    \[
   e^{\epsilon}v^{\intercal}Yv \geq v^{\intercal}Xv\geq e^{-\epsilon}v^{\intercal}Yv
    \]
    \item If $e^{\epsilon}v^{\intercal}Xv\geq v^{\intercal}Yv\geq e^{-\epsilon}v^{\intercal}Xv$ and $c\geq 0$ then 
    \[
    c\cdot e^{\epsilon}v^{\intercal}Xv\geq c\cdot v^{\intercal}Yv\geq c\cdot e^{-\epsilon}v^{\intercal}Xv
    \]
    which implies $c\cdot X\approx_{\epsilon}c\cdot Y$
    \item Suppose $e^{\epsilon}v^{\intercal}Xv\geq v^{\intercal}Yv\geq e^{-\epsilon}v^{\intercal}Xv$. Then 
    \begin{align*}
        e^{\epsilon}v^{\intercal}(X+W)v &=e^{\epsilon}v^{\intercal}Xv+e^{\epsilon}v^{\intercal}Wv\\
        &\geq v^{\intercal}Yv +v^{\intercal}Wv\\
        &=v^{\intercal}(Y+W)v
    \end{align*}
    and the case with $e^{-\epsilon}$ is symmetric.
    \item Suppose $e^{\epsilon}v^{\intercal}Xv\geq v^{\intercal}Yv\geq e^{-\epsilon}v^{\intercal}Xv$ and $e^{\epsilon}v^{\intercal}Wv\geq v^{\intercal}Zv\geq e^{-\epsilon}v^{\intercal}Wv$. Then 
    \begin{align*}
        e^{\epsilon}v^{\intercal}(X+W)v &=e^{\epsilon}v^{\intercal}Xv+e^{\epsilon}v^{\intercal}Wv\\
        &\geq v^{\intercal}Yv +v^{\intercal}Zv\\
        &=v^{\intercal}(Y+Z)v
    \end{align*}
    and the case with $e^{-\epsilon}$ is symmetric.
    \item Suppose $e^{\epsilon_1}v^{\intercal}Xv\geq v^{\intercal}Yv\geq e^{-\epsilon_1}v^{\intercal}Xv$ and $e^{\epsilon_2}v^{\intercal}Yv\geq v^{\intercal}Zv\geq e^{-\epsilon_2}v^{\intercal}Yv$. Then
    \begin{align*}
        e^{\epsilon_1+\epsilon_2}v^{\intercal}Xv &=e^{\epsilon_2}\cdot\left(e^{\epsilon_1}v^{\intercal}Xv\right)\\
        &\geq e^{\epsilon_2}\cdot\left(v^{\intercal}Yv\right)\\
        &\geq v^{\intercal}Zv
    \end{align*}
    and the case with $e^{-\epsilon_1-\epsilon_2}$ is symmetric.
     \item Suppose that for all $v\in\mathbb{R}$, $e^{\epsilon}v^{\intercal}Xv\geq v^{\intercal}Yv\geq e^{-\epsilon}v^{\intercal}Xv$ and let $M$ be a matrix. Let $y=Mv$. Then
    \begin{align*}
    e^{\epsilon}v^{\intercal}M^{\intercal}XMv&=e^{\epsilon}y^{\intercal}Xy\\ 
    &\geq y^{\intercal}Yy\\
    &=v^{\intercal}M^{\intercal}YMv
    \end{align*}
    and the other inequality follows by symmetry.
    \item  Since $X$ and $Y$ have the same kernel, it suffices to prove the claim restricted to the orthogonal complement of the kernel and so we may assume without loss of generality that $X$ and $Y$ are positive definite and show that $X\approx_{\epsilon}Y\implies X^{-1}\approx_{\epsilon}Y^{-1}$. We will show that in general for positive definite matrices $A$ and $B$ $A\preceq B$ implies $B^{-1}\preceq A^{-1}$. The claim then follows by taking $B=e^{\epsilon}X$ and $A=Y$ and then $B=Y$ and $A=e^{-\epsilon}X$.  

Since $B$ is a positive definite symmetric matrix, it has a positive definite square root $B^{1/2}$. Let $M=B^{-1/2}AB^{-1/2}$ and note that $M^{-1}=B^{1/2}A^{-1}B^{1/2}$. Then $M\preceq I$ if and only if $M^{-1}\succeq I$ because both correspond to $M$ having eigenvalues at most 1. Applying Part 6 to $A\preceq B$ by multiplying the left and right by $B^{-1/2}$ gives $M\preceq I$ and hence $M^{-1}\succeq I$. Applying Part 6 again with $B^{-1/2}$ gives $B^{-1/2}M^{-1}B^{-1/2} \succeq B^{-1/2}B^{-1/2}$, which implies $A^{-1}\succeq B^{-1}$.

\item Suppose $X$ and $Y$ are $n\times n$ matrices and $I$ is $m\times m$. Fix $v\in \mathbb{R}^{m\cdot n}$ and let $v_1$ be a vector of the first $n$ entries of $v$, $v_2$ be a vector of the next $n$ entries of $v$ etc. Observe that
\begin{align*}
v^{\intercal}(I\otimes X)v &= v_1^{\intercal}X v_1 +v_2^{\intercal}Xv_2 +\ldots+v_m^{\intercal}X v_m\\
&\leq e^{\epsilon} \cdot \left(v_1^{\intercal}Y v_1 +\ldots+v_m^{\intercal}Y v_m\right)\\ 
&=e^{\epsilon}\cdot v^{\intercal}(I \otimes Y) v
\end{align*}
and the other inequality follows by symmetry.
\end{enumerate}
\end{section}

\begin{section}{Proof of Proposition \ref{prop:identity}}
\label{app:identity}
\begin{proof}
We'll show that the two sides have the same eigenvectors and eigenvalues. Let $v$ be any eigenvector of $L$ with eigenvalue $\gamma\not=0$ (so $v\perp\vec{1}$). Then $v$ is also an eigenvector of $I,J,M$, and $I-M^2$ with eigenvalues 1,0,$\lambda=1-\gamma$, and $1-\lambda^2\not=0$, respectively. So when we apply the right-hand side to $v$, we get
\begin{align*}
\frac{1}{2}\left(1-0+(1+\lambda)\cdot\frac{1}{1-\lambda^2}\cdot(1+\lambda)\right)\cdot v &=\frac{1}{1-\lambda}\cdot v\\
&=\frac{1}{\gamma}\cdot v\\
&=L^{+}v
\end{align*}
If $v$ has eigenvalue 0 under $L$, then $v$ is a multiple of the all $1$'s vector (since $G$ is connected) and has eigenvalues 1,1,1, and 0 under $I,J, M$, and $I-M^2$, respectively. So applying the right-hand side to $v$ gives 
\begin{align*}
\frac{1}{2}\left(1-1+(1+1)\cdot 0\cdot(1+1)\right)\cdot v&=0\\
&=L^{+}v
\end{align*}
\end{proof}
\end{section}
\begin{section}{Proof of Lemma \ref{lem:equivalence}}
\label{app:equivalence}
\begin{proof}
Note that a $1/2$-lazy transition matrix can be written as $M=\frac{1}{2}(T+I)$ where $T$ is another transition matrix. Suppose $\lambda(G)\leq \lambda$. This means that for all $v\perp \vec{1}$ we have
\begin{align*}
v^{\intercal}Lv&=v^{\intercal}Iv-v^{\intercal}Mv\\
&\geq (1-\lambda)v^{\intercal}v\\
&=e^{-\epsilon}\cdot v^{\intercal}(I-J)v
\end{align*}

We also have
\begin{align*}
v^{\intercal}Lv &=\frac{1}{2}(v^{\intercal}Iv-v^{\intercal}Tv)\\
&\leq \frac{1}{2}(v^{\intercal}v+v^{\intercal}v)\\
&\leq e^{\epsilon}v^{\intercal}(I-J)v\\
\end{align*}
In addition if $v=\vec{1}$, we have 
\begin{align*}
v^{\intercal}Lv&=0\\
&=v^{\intercal}(I-J)v
\end{align*}
It follows that $L\approx_{\epsilon}I-J$

Conversely, suppose that $L\approx_{\epsilon}I-J$. Then for all $v\perp \vec{1}$ we have
\begin{align*}
v^{\intercal}Lv&\geq e^{-\epsilon}v^{\intercal}(I-J)v\\
&=(1-\lambda)v^{\intercal}v
\end{align*}
This means that the nontrivial eigenvalues of $L$ are at least $1-\lambda$, i.e. all the nontrivial eigenvalues of $M$ are at most $\lambda$. Since $G$ is $1/2$-lazy all the eigenvalues are nonnegative and hence $\lambda(G)\leq \lambda$.
\end{proof}
\end{section}
\begin{section}{Proof of Lemma \ref{lem:pseudoinv_implies_sol}}
\label{app:pseudoinv_implies_sol}
\begin{proof}
Let $x^{*}=L^{+}b$, which is a solution to $Lx=b$ by Proposition \ref{prop:pseudoinversesolution}. We want to show that
\[
\|x^{*}-\tilde{x}\|^{2}_{L}\leq 2\epsilon\cdot\|x^{*}\|_{L}^{2}
\]

Note that 
\begin{align}
\label{eq:Lnorm}
\|x^{*}\|_{L}^{2}&= x^{*^{\intercal}}Lx^{*} \nonumber\\ 
&=(L^{+}b)^{\intercal}L(L^{+}b)\nonumber\\
&=b^{\intercal}L^{+}LL^{+}b\nonumber\\
&=b^{\intercal}L^{+}b 
\end{align}
By Proposition \ref{prop:psdfacts} Part 7, $\tilde{L}\approx_{\epsilon}L$ so we can write
\begin{align*}
\|x^{*}-\tilde{x}\|^{2}_{L}&=x^{*^{\intercal}}Lx^{*}+\tilde{x}^{\intercal}L\tilde{x}-2\tilde{x}^{\intercal}Lx^{*}\\
&\leq b^{\intercal}L^{+}b+e^{\epsilon}\tilde{x}^{\intercal}\tilde{L}\tilde{x}-2(\tilde{L}^{+}b)^{\intercal}b\\
&=b^{\intercal}L^{+}b+(e^{\epsilon}-2)b^{\intercal}\tilde{L}^{+}b
\end{align*}
Since $\epsilon\leq \ln(2)$, $(e^{\epsilon}-2)$ is nonpositive. So
\begin{align*}
\|x^{*}-\tilde{x}\|_{L}^{2}&\leq b^{\intercal}L^{+}b+(e^{\epsilon}-2)b^{\intercal}\tilde{L}^{+}b\\
&\leq b^{\intercal}L^{+}b+e^{-\epsilon}(e^{\epsilon}-2)b^{\intercal}L^{+}b\\
&=2(1-e^{-\epsilon})b^{\intercal}L^{+}b\\
&\leq 2\epsilon\|x^{*}\|_{L}^{2}
\end{align*}
where the last line comes from \ref{eq:Lnorm}.
\end{proof}
\end{section}
\end{document}